%% file: paper.tex
\documentclass[12pt]{iopart}

\usepackage{iopams} 

\expandafter\let\csname equation*\endcsname\relax
\expandafter\let\csname endequation*\endcsname\relax 

\usepackage{color}
\usepackage{amssymb}
\usepackage{amsthm}
\usepackage{MnSymbol}
\usepackage{graphicx}
\usepackage{dcolumn}
\usepackage{bm}
\usepackage{graphicx}
\usepackage{epsfig}
\usepackage{mathrsfs}

\usepackage[all,2cell]{xy}

\newcommand{\be}{\begin{equation}}
\newcommand{\ee}{\end{equation}}
\newcommand{\bea}{\begin{eqnarray}}
\newcommand{\eea}{\end{eqnarray}}
\newcommand{\bes}{\begin{subequations}\bea}
\newcommand{\ees}{\eea\end{subequations}}
\newcommand{\ba}{\begin{array}}
\newcommand{\ea}{\end{array}}
\newcommand{\bos}[1] {\boldsymbol{#1}}
\newcommand{\prl}{[}
\newcommand{\prr}{]}
\newcommand{\flux}{\varphi}
\newcommand{\fluxr}{{\varphi'}}

\newtheorem{definition}{Definition}
\newtheorem{theorem}{Theorem}

\begin{document}
 
\title{BEST statistics of Markovian fluxes: \\ a tale of Eulerian tours and Fermionic ghosts}

\author{Matteo Polettini$^1$}

\address{$^1$ Facult\'e des Sciences, de la Technologie et de la Communication
162 A, avenue de la Fa\"\i{}encerie, L-1511 Luxembourg (Grand Duchy of Luxembourg) }

\ead{matteo.polettini@uni.lu}

\begin{abstract} 
We provide an exact expression for the statistics of the fluxes of Markov jump processes at all times, improving on asymptotic results from large deviation theory. The main ingredient is a  generalization of the BEST theorem in enumeratoric graph theory to Eulerian tours with open ends. In the long-time limit we reobtain Sanov's theorem for Markov processes, which expresses the exponential suppression of fluctuations in terms of relative entropy. The finite-time power-law term, increasingly important with the system size, is a spanning-tree determinant that, by introducing Grassmann variables,  can be absorbed into the effective Lagrangian of a Fermionic ghost field on a metric space, coupled to a gauge potential. With reference to concepts in  nonequilibrium stochastic thermodynamics, the metric is related to the dynamical activity that measures net communication between states, and the connection is made to a previous gauge theory for diffusion processes.
\end{abstract}



\vspace{2pc}




\input{body}

\input{biblio}
\end{document}

%% file: body.tex
which will be usefulas \section{Introduction}


One of the lesser-known intuitions of Einstein,  which he applied the critical opalescence of mixtures, was the idea of ``inverting Boltzmann's formula'' to obtain an expression for the probability of a fluctuating variable in terms of entropy  \cite{einstein}. A similar mechanism underlies the important result by Sanov \cite{sanov} in what came to be known as the theory of large deviations (LDT): relative entropy describes the asymptotic behavior of the probability of an increasing number of independent and identically distributed variables. The theory has since grown to embrace Markov processes and beyond.

Recent years have witnessed a burst of interest in the characterization of the thermodynamics of nonequilibrium systems by means of classical and quantum Markov processes \cite{esposito}.  The celebrated Fluctuation Theorem \cite{crooks,kurchan,maes,lebowitz,gaspard,schutz,polettiniFT}, which expresses a symmetry of the large fluctuations of thermodynamic observables, paved the way for a thorough application of the full spectrum of methods and results from LDT \cite{touchette1}. Here,  ``large'' often indicates the long time limit, regulating the exposure of the system to external agents, whereas in equilibrium statistical mechanics one usually refers to size.  This revival brought about computational \cite{giardina,tailleur1} and theoretical \cite{nemoto,koza,wachtel,plenio} techniques for a more efficient computation of rate functions and cumulants of time-extensive observables, a better understanding of the analogies and differences between equilibrium and nonequilibrium systems \cite{lecomte2,andrieux,chetrite1,tailleur2}, of phase transitions along most probable paths \cite{hurtado,bouchet}, of the hydrodynamic limit \cite{jona} and of exotic properties of time-intensive observables (e.g. the efficiency) \cite{verley,polettinieff}.

Large deviation principles meliorate the central limit theorem. Still there are further margins for improvement in considering finite-time effects. However, unfortunately, obtaining rate functions in LDT is often prohibitive, but for very simple or special systems \cite{derrida,lazarescu2}. The reason for this is that macroscopic observables of interest are aggregate quantities to whom contribute many, more fundamental, correlated processes. General expressions are only available when one refines the description to all empirical fluxes and densities populating the system \cite{maes2}.

Indeed, as we will prove in this paper, for Markov processes on a discrete state space, this level of description (sometimes called ``level 2.5'') is so complete that it actually allows for the derivation of the full joint probability of the fluxes and of the empirical measure. The finite-time corrections are accounted for by a combinatorial factor counting Eulerian tours or, equivalently, spanning trees on a suitably defined digraph. Employing the same tools as in a recently proposed Fermionic field theory for trees and forests \cite{caracciolo,abdesselam}, we express the new finite-time term as an effective Lagrangian of a ghost field coupled to a gauge potential, on a metric. The metric depends on the dynamical activity of the system \cite{maes2}, and a connection is made to a previously proposed gauge theory of nonequilibrium thermodynamics by the Author \cite{polettini1}.  Finally, we re-derive large deviation rate functions by a rather different procedure than is usually employed. This direct evaluation allows to appreciate a subtlety related to the large deviations of the statistics of the total number of jumps, which is a stochastic variable affecting the set of measurable events along a trajectory.

To facilitate the reading, we will first provide all results in Sec.\,\ref{restless} in the simpler setting of a restless, discrete-time Markov chain, reobtaining Sanov's theorem for Markov chains and discussing finite-time corrections to small fluctuations. In Sec.\,\ref{fermionic} we discuss the ghost Lagrangian, and in Sec.\,\ref{continuous} we generalize to Markov jump processes in continuous time. When considering large deviations for the joint fluxes and empirical densities, in \S\,\ref{largecont}, we discuss the subtle role played by the total number of jumps. We draw conclusions in Sec.\,\ref{conclusions}.

\section{The simplest case: restless Markov chains}

\label{restless}

\subsection{Setup}

Consider a graph  $G =(X,E)$ consisting of a finite number of vertices $x,y,\ldots \in X$ connected by edges $xy \in E$. We suppose the graph has no loops (edges that connect a state to itself) nor multiple edges. The graph is undirected (arrows do not have a preferential direction), but eventually an arbitrary orientation can be established, which will be useful in Sec.\,\ref{fermionic}.

On such state space, a Markovian walker performs a path $\bos{x} = (x_0,\ldots,x_N)$,  departing from $x_0$ and at intervals of unit time performing transitions to neighboring vertices $x_i$. If edge $xy$ exists, then both the probability $\pi_{xy}$ of jumping from $y$ to $x$ and its reverse $\pi_{yx}$ are nonvanishing, while we assume null probability of resting at states, $\pi_{xx} = 0, \forall x$. The probability of a path $\bos{x}$ is given by
\bea
P\prl \bos{x} \prr = \prod_{i=0}^{N-1} \pi_{x_{i+1}x_i}. \label{eq:pathprob}
\eea
Introducing the transition matrix $\Pi = (\pi_{xy})_{x,y}$, 
the probability $\rho_i(x)$ of being at $x$ at time $i$ satisfies $\bos{\rho}_{i+1} = \Pi \bos{\rho}_i$, given $\rho_0(x) = \delta_{x,x_0}$. Under the above assumptions there exists a unique invariant probability $\bos{\rho}$ that satisfies $\Pi \bos{\rho} = \bos{\rho} = \bos{\rho}_{N \to +\infty}$. 

The objects of crucial interest in this paper are the {\it fluxes}, stochastic variables that count the number of transitions between states along a path:
\bea
\flux_{xy}\prl \bos{x} \prr & := & \sum_{i = 0}^{N-1} \delta_{x_{i+1},x} \, \delta_{x_i,y} . \label{eq:fluxes} 
\eea
Flux $\flux_{xy}$ increases by $+1$ every time transition $x \gets y$ occurs. We collectively denote the fluxes $\bos{\flux} = (\flux_{xy})_{x,y}$. Notice that
$\sum_{x,y} \flux_{yx}\prl \bos{x} \prr = N$. It is often physically meaningful to focus on the symmetric and antisymmetric parts of the fluxes, called respectively {\it activities} and {\it currents}:
\bea
\flux_{xy}^{\pm} = \frac{1}{2}\left( \flux_{xy} \pm \flux_{yx} \right). \label{eq:actcur}
\eea

Physical observables typically have the form $\sum_{x,y} \mathrm{obs}_{xy} \flux_{xy}\prl \bos{x} \prr$. In particular in the context of nonequilibrium statistical mechanics, thermodynamic observables like the entropy production, the heat fluxes etc. are antisymmetric, $\mathrm{obs}_{xy} = - \mathrm{obs}_{yx} $, hence their statistics only depends on the currents. However, while such antisymmetric observables suffice to capture the physics near equilibrium, e.g. via the fluctuation-dissipation relation and the minimum entropy production principle \cite{polettiniminEP}, it is now understood that out of equilibrium also symmetric observables related to dynamical effects enter into play in a crucial manner. For example, they are necessary to extend the fluctuation-dissipation relation out of nonequilibrium steady states \cite{maes1}.

Importantly, any such kind of observable only depends on the fluxes and not on the precise trajectory that generates that configuration of fluxes. Then a preliminary step is to obtain the probability of the fluxes
\bea
P\prl \bos{\flux} \prr  & = &  \sum_{\bos{x}} \bos{\delta}_{\bos{\flux},\bos{\flux}\prl \bos{x} \prr} P\prl \bos{x} \prr \\
& = &  \int \mathcal{D}(\ldots) \, e^{\mathcal{L}(\ldots)} \label{eq:lag}
\eea
where $\bos{\delta}$ is a product of Kronecker deltas, one per flux, and the second expression is a sort of yet-to-determine representation in terms of an effective Lagrangian. Asymptotic expressions for $N\to + \infty$ are well-known. The task of this paper is to exactly evaluate $P\prl \bos{\flux} \prr$ at all times and to make sense of the second expression. The problem of further marginalizing to evaluate the statistics of actual observables is intractable in full generality. 

\subsection{BEST statistics of the fluxes}

Given Eq.\,(\ref{eq:fluxes}), we can recast the path probability Eq.\,(\ref{eq:pathprob}) as
\bea
P\prl \bos{x} \prr = \prod_{x,y} \pi_{xy}^{\flux_{xy}} . \label{eq:path}
\eea
We can then marginalize for the fluxes to obtain
\bea
P\prl \bos{\flux} \prr = \varepsilon\prl \bos{\flux} \prr \exp \left(\sum_{x,y} \flux_{xy} \ln \pi_{xy}\right),  \label{eq:Pphi}
\eea
where
\bea
\varepsilon\prl \bos{\flux} \prr = \sum_{\bos{x}} \bos{\delta}_{\bos{\flux},\bos{\flux}\prl \bos{x} \prr} 
\eea
is the number of distinguishable paths that perform $\flux_{xy}$ times transition $x \gets y$, for all transitions.

The key step in our treatment is the determination of this number, basing on a generalization of the so-called BEST\footnote{After de Bruijn, van Aardenne-Ehrenfest, Smith and Tutte.} theorem in enumeratoric graph theory. Let us describe the general procedure; in \ref{app3} we provide an example of what follows. Corresponding to a set of fluxes $\bos{\flux} = (\flux_{xy})_{x,y}$ we can draw a {\it digraph} $(X,\bos{\flux})$ with $\flux_{xy}$ arrows $y \to x$. In this case, the arrows are considered to be directed. The {\it influx} and the {\it outflux} (in graph-theoretic jargon, {\it in-degree} and {\it out-degree}) at vertex $x$ are respectively defined  as the number of transitions towards and out of $x$,
\bea
\flux^{\mathrm{in}}_x\prl \bos{x} \prr := \sum_{y} \flux_{xy}\prl \bos{x} \prr \label{eq:degree}, \qquad 
\flux^{\mathrm{\,out}}_x\prl \bos{x} \prr := \sum_{y} \flux_{yx}\prl \bos{x} \prr. 
\eea
Since a Markovian path can be drawn along the edges of the graph by ``never lifting the pencil'', then the influxes and the outfluxes of all vertices but the initial and the final ones balance each other, and in general we have
\bea
\flux^{\mathrm{in}}_x\prl \bos{x} \prr - \flux^{\mathrm{\,out}}_x\prl \bos{x} \prr = \delta_{x_N,x} - \delta_{x_0,x}.  \label{eq:inout}
\eea
When $x_0 = x_N$ the digraph is said to be {\it balanced}. Historically the first theorem in graph theory, motivated by the problem whether it was possible to traverse the seven K\"onisberg bridges exactly once to arrive at the place of departure, was Euler's proof that a digraph is balanced if and only if there exists a cyclic path, called {\it Eulerian tour}, that traverses each arrow exactly once. Importantly for what follows, notice that in an Eulerian tour on a balanced digraph each arrow is accounted distinctly, while the paths we are considering are blind to the arrow's identity. The BEST theorem in enumeratoric graph theory counts the number of different Eulerian tours on a balanced digraph (see \cite{stanley}, \S 5.6).

Generalizing, we will call a digraph whose in- and out- degrees obey Eq.\,(\ref{eq:inout}) a {\it quasi-balanced digraph}. Vertex $x_0$ is called the {\it source} and vertex $x_N$ the {\it sink}. In \ref{appA} it is proven that a digraph is quasi-balanced if and only if it is possible to draw a {\it quasi-Eulerian tour} going from $x_0$ to $x_N$ and such that, along the lines of the BEST theorem, the number of independent quasi-Eulerian tours is given by
\bea
E = T_{x_N} \prod_x \big(\flux^{\mathrm{\,out}}_{x} -1 + \delta_{x,x_N}\big)!,
\eea
where $T_{x}$ is the number of {\it arborescenses}, or {\it directed rooted spanning trees} pointing at $x$. This number can be expressed as a matrix determinant \cite{BEST}. To this purpose, we define the $|X_N| \times |X_N|$ digraph Laplacian matrix (also called Kirchhoff-Tutte matrix) 
\bea
\Delta_{x,y} := \left\{ \ba{ll} - \flux_{xy},  & x \neq y \\ \flux^{\mathrm{\,out}}_y, & x = y \ea \right., \quad \forall x,y \in X_N
\eea
where $X_N$ is the set of states that at time $N$ have been touched at least once by the trajectory, that is such that $\flux^{\mathrm{in}}_x + \flux^{\mathrm{\,out}}_x > 0$. Notice that $\sum_{x} \Delta_{x,y}  = 0$, hence $\det \Delta = 0$.
Notoriously, the matrix-tree theorem in algebraic graph theory establishes that the cofators of the $x$-th column of the Laplacian matrix give the number of oriented spanning trees with root $x$. In particular, for a digraph we have that 
\bea
T_{x_N} = \det \Delta_{x_N} 
\eea
where $\Delta_{x_N}$ is obtained by removing from $\Delta$ the row and column corresponding to $x_N$. Finally, to obtain the number of distinguishable paths we should identify all permutations of edges in the same direction, introducing a Gibbs-type factor. We finally obtain
\bea
\varepsilon\prl \bos{\flux} \prr = \frac{E\prl \bos{\flux} \prr}{\prod_{x,y} \flux_{xy}!} = \det \Delta_{x_N}\prl \bos{\flux} \prr \, \frac{\prod_x \left(\flux^{\mathrm{\,out}}_{x} -1 + \delta_{x,x_N}\right)!}{\prod_{x,y} \flux_{xy}!},
\eea
where we made explicit the dependencies on the fluxes. This, together with Eq.\,(\ref{eq:Pphi}), constitutes the main result of this paper. Putting them together, 
assuming that the fluxes 
$\bos{\flux}$ are large enough so that we can approximate $\flux^{\mathrm{in}}_{x} \sim \flux^{\mathrm{\,out}}_{x} = \flux_x$ and use Stirling's formula $\flux_{xy}! \sim \exp ( \flux_{xy} \ln \flux_{xy})$, we obtain 
\bea
P\prl \bos{\flux} \prr = \det \Delta_{x_N}\prl \bos{\flux} \prr \; e^{-I\prl \bos{\flux} \prr} \label{eq:PPP}
\eea
where
\bea
I\prl \bos{\flux} \prr := \sum_{x,y} \flux_{xy} \ln \frac{\flux_{xy}}{\pi_{xy}\flux_y}.
\eea
Interestingly, the digraph Laplacian above can be seen as the generator of a continuous-time Markov chain for which a particular configuration of fluxes is typical. The same prefactor has also been recently derived in Ref.\,\cite{gatver} by completely different methods.

\subsection{Large deviations}

It will now be convenient to scale quantities in the extensive parameter $N$. Let us define the {\it flux ratio} $\bos{\fluxr}$ as the probability that, if a transition occurs, that transition is $x \gets y$, 
\bea
\fluxr_{xy} := \frac{\flux_{xy}}{\sum_{x,y}\flux_{xy}} = \frac{\flux_{xy}}{N}.
\eea
Notice that the flux ratio has the meaning of a flux rate, in discrete time.
Letting $N$ be large but finite (so that the flux ratios remain discrete variables), the probability of the flux ratios is given by $P'\prl \bos{\fluxr} \prr = P\prl N\bos{\fluxr}\prr$ and we define
\bea
I'\prl \bos{\fluxr} \prr := \frac{I(N\bos{\fluxr})}{N} = D(\bos{\fluxr}\Vert \bos{\pi}^\fluxr)
\eea
where $\pi^\fluxr_{xy} := \pi_{xy}\fluxr_y$, and $D(\,\cdot\,\Vert\,\cdot\,)$ is the Kullback-Leibler divergence. Notice that in  Eq.\,(\ref{eq:PPP}) while this entropic term decreases exponentially, the determinant increases as $N^{|X_N|}$, but since $|X_N|$ is at most $|X|$ this term is subdominant \footnote{In the limit $N \to \infty$ the flux rates become continuous variables whose probability density scales like $N^{2|E|}$, which again provides a subdominant term.}:
\bea
- \frac{\log P'_N\prl \bos{\fluxr} \prr}{N} = D(\bos{\fluxr}\Vert \bos{\pi}^\fluxr) + |X_N| \frac{\ln N}{N} + \frac{\ln \det \Delta_{x_N}\prl \bos{\fluxr} \prr}{N}.
\eea
Then $I' = - \lim_{N \to +\infty} \log P'_N /N$ is the large deviation rate function of the process. We thus recover, by a different procedure than usually employed, the Sanov theorem for Markov chains \cite{sanov,kesidis,fayolle}.

The most probable configuration of fluxes (marked $^\ast$) is the minimum of the rate function, found by solving $D(\bos{\fluxr}^\ast\Vert \bos{\pi}^{\fluxr^\ast}) = 0$, which yields $\bos{\fluxr}^\ast = \bos{\pi}^{\fluxr^\ast}$. This solution is indeed unique, as by summing over $y$ we obtain $\fluxr_x = \sum_y \pi_{xy} \fluxr_y$, which together with $\sum_x \fluxr_x = 1$ grants that $\bos{\fluxr} = \bos{\rho}$ is the unique invariant state and $\bos{\fluxr}^\ast =  \bos{\pi}^{\rho}$.

\subsection{Fluctuation relation}

Another testing ground for our formula is the fluctuation relation (see Refs.\,\cite{polettiniFT,polettiniGRAPH} for a general treatment of the fluctuation relation for the currents). Let us reverse the direction of all the fluxes by defining
\bea
\flux^\dagger_{xy} = \flux_{yx} .
\eea
Notice that in the reversed digraph $(X,\bos{\flux}^\dagger)$, $x_0$ plays the role of the sink and $x_N$ of the source. It is obvious that there is a one-to-one correspondence between the number of Eulerian tours from $x_0$ to $x_N$ in $(X,\bos{\flux})$ and that of Eulerian tours from $x_N$ to $x_0$ in $(X,\bos{\flux}^\dagger)$. Then we easily obtain that
\bea
\frac{P\prl \bos{\flux} \prr}{P\prl\bos{\flux}^\dagger\prr} =  \exp \sum_{x,y} \flux^-_{xy} \ln \frac{\pi_{xy}}{\pi_{yx}}.
\eea
Notice that only currents enter this expression. Hence one can finally marginalize to obtain a finite-time fluctuation relation for the currents. The functional at exponent in the right-hand side is called the {\it entropy production} or Lebowitz-Spohn action functional.  

Finally, the fact that there is a one-to-one correspondence between Eulerian tours on $(X,\bos{\flux})$ and of reversed Eulerian tours on $(X,\bos{\flux}^\dagger)$ implies that the number $T_{x_N}$ of oriented  spanning trees with root at $x_N$ in $(X,\bos{\flux})$ and that $T^\dagger_{x_0}$ of spanning trees with root at $x_0$ in $(X,\bos{\flux}^\dagger)$ are related by
$T_{x_N}/T^\dagger_{x_0} = \flux^{\mathrm{\,out}}_{x_N} / \flux^{\mathrm{\,out}}_{x_0}$. This is far from obvious as the shapes of these spanning trees might be very different. 

\subsection{Small fluctuations} 

In this section we consider small fluctuations out of the most probable solution, at sufficiently large but finite $N$, evaluating the leading correction coming from the finite-time term. We set
\bea
\bos{\flux} = N \left(\bos{\pi}^\rho  + \frac{1}{\sqrt{N}} \, \delta\bos{\flux} \right)
\eea
where $N^{-1/2}\delta\bos{\flux} \ll \bos{\pi}^\rho$ for all edges. Here, it is assumed that ``reasonably probable'' fluctuations scale with $\sqrt{N}$ while most probable values scale with $N$.  
The leading term in the large deviation rate function is of order $O(1)$. In fact, 
expanding in Taylor series one obtains
\bea
I\prl \bos{\flux} \prr
& \approx &  \frac{1}{N} \sum_{x,y} \frac{ \left(\flux_{xy}   - \pi_{xy} \flux_{y} \right)^2 }{\pi_{xy} \rho_y}   \\
& \approx &  \sum_{x,y} \frac{\delta \flux_{xy} }{\pi_{xy} \rho_y}  \left(\delta \flux_{xy}   - \pi_{xy} \delta \flux_{y} \right)
\eea
Let us now show that finite-time corrections are of order $N^{-1/2}$. Using the log-trace formula for the determinant, we have
\begin{multline}
\det \Delta_{x_N}\prl \bos{\flux} \prr \approx N^{|X_N|}  \det \Delta_{x_N} [\bos{\pi}^\rho] \,
\exp \tr \log \left(\mathrm{Id}_{|X_N|} +  N^{-1/2} \Delta_{x_N} [\bos{\pi}^\rho]^{-1}  \Delta_{x_N} [\delta\bos{\flux}] \right). 
\end{multline}
As regards the prefactor, notice that $\Delta_{x_N} \prl \bos{\pi}^\rho \prr = \Delta_{x_N} \prl \bos{\pi} \prr \, \mathrm{diag}(\rho_x)_{x\neq x_N}$. Taking the determinant
\bea
\det \Delta_{x_N} \prl \bos{\pi}^\rho \prr = \det \Delta_{x_N} \prl \bos{\pi} \prr \prod_{x \neq x_N} \rho_x
\eea
we notice that the determinant of the minor of the Laplacian matrix containing the transition probabilities appears. It is well known that, by the matrix-tree theorem for Markov chains \cite{gaveau}, we have $\rho_{x_N} =  \Delta_{x_N} \prl \bos{\pi} \prr / T(\bos{\pi})$ with the {\it spanning-tree polynomial}
\bea
T(\bos{\pi}) = \sum_x \sum_{\mathcal{T}_x} \prod_{e \in \mathcal{T}_x} \pi_e,
\eea
where the sum is over all possible roots $x$ and over oriented spanning trees $\mathcal{T}_x$ with root $x$, viz. maximal sets of edges of the graph that contain no cycles and such that there exists a unique directed path from each vertex to $x$. The product is over all edges belonging to the spanning tree.

Putting all together we obtain 
\bea
P[\bos{\flux}]= N^{|X_N|} T(\bos{\pi}) \prod_{x} \rho_x \; \exp \left\{ -  I[\bos{\flux}] + \delta I [\bos{\flux}] \right\}
\eea
where the correction to the rate function is given by
\bea
\delta I[\bos{\flux}] = N^{-1/2} \tr \left\{\Delta_{x_N} [\bos{\pi}^\rho] ^{-1}\Delta_{x_N} [\delta\bos{\flux}] \right\}.
\eea
Hence the theory allows to obtain a simple expression for the deviations of order $N^{-1/2}$ to the large deviation rate function, when considering a  reasonably probable fluctuation of size $N^{-1/2}$. The correction might be simple to compute and contract for the macroscopic currents in specific systems such as exclusion processes and the zero-range process. 

\section{A ghost field theory for the fluxes}

\label{fermionic}

\subsection{Ghosts}

While the dominant behavior at large times is entropic, it is tempting to associate some form of energy to finite-time corrections by incorporating the determinant $T_{x_N} = \det \Delta_{x_N}$ into an overall Lagrangian.

A standard procedure to exponentiate determinants is to introduce anticommuting Grassmann variables $\bos{\psi} = (\psi_x)_{x}$ and $\bar{\bos{\psi}} = (\bar{\psi}_x)_{x}$ at all vertices of the graph. To this purpose we notice that
$\det \Delta_{x_N} = \det \Delta'$, 
where $\Delta'$ is obtained by adding $+1$ in the row and column of $\Delta$ corresponding to $x_N$, $\Delta'_{xy} = \Delta_{xy} + \delta_{x,x_N}\delta_{y,x_N}$, as can be easily seen by expanding with Laplace's formula and keeping into account that $\det \Delta = 0$.
Then
\bea
T_{x_N}  = \det \Delta' =  \int \mathcal{D}(\bos{\psi},\bar{\bos{\psi}}) \; e^{\bar{\bos{\psi}} \cdot \Delta' \bos{\psi}}
\eea
where the integration over the Grassmann variables follows the rules of Berezin's calculus, that we do not discuss here (see Ref. \cite{abdesselam} for a review with application to matrix-tree theorems). The second expression in Eq.\,(\ref{eq:lag}) is then obtained by identifying the Lagrangian\footnote{An alternative, and possibly more elegant expression, is to write the theory as an expectation value \cite{abdesselam}
$ P= \int \mathcal{D}(\bos{\psi},\bar{\bos{\psi}}) \bar{\psi}_{x_N} \psi_{x_N} \, e^{\mathcal{L}_0}$ with respect to the Lagrangian  $\mathcal{L}_0(\bar{\bos{\psi}},\bos{\psi},\bos{\flux}) = \bar{\bos{\psi}} \cdot \Delta\prl \bos{\flux} \prr \bos{\psi} - I\prl \bos{\flux} \prr$}
\bea
\mathcal{L}(\bar{\bos{\psi}},\bos{\psi},\bos{\flux}) = \bar{\bos{\psi}} \cdot \Delta\prl \bos{\flux} \prr \bos{\psi} + \bar{\psi}_{x_n}\psi_{x_n} - I\prl \bos{\flux} \prr
\eea

At finite times there is a competition between an entropic term and a kinetic term of a discretized {\it ghost field} of the kind introduced by Faddeev and Popov in Yang-Mills quantum theories.

\subsection{Metric and gauge potential}

The analogy to ghost fields can be made stronger. The scope of this section is to recast this expression as
\bea
\bar{\bos{\psi}} \cdot \Delta \bos{\psi} = (\partial \bar{\bos{\psi}},D \bos{\psi})_\Phi \label{eq:gaugemetric} 
\eea
with $\partial$ an appropriate divergence, $D = \partial + A$ a covariant derivative, $A$ a gauge potential, $\Phi$ a metric.

First, the proper concept of divergence on graphs is defined on the oriented graph $G_{\prec} = (X,E_\prec)$,  where we introduce an arbitrary order relation between states $x \prec y$ which establishes an orientation of each edge, that we also denote $x \prec y$. We then introduce the incidence matrix $\partial$ with entries
\bea
\partial_{y \prec z,x} := \left\{ \ba{ll} +1, & \mathrm{if}~ y = x \\ -1, & \mathrm{if}~z = x \\ 0 & \mathrm{otherwise} \ea \right. , 
\eea
that fully describes the topology of the graph. Indeed, the incidence matrix is the discrete version of the divergence, when acting on scalar fields (defined on the vertex set) on the right, and of the gradient, when acting on vector fields (see Ref.\,\cite{ddg} for elements of discrete differential geometry).

The latter vector fields are defined as antisymmetric edge observables.  Hence we move from the fluxes to the activities and currents introduced in Eq.\,(\ref{eq:actcur}), and let $\Delta' = \Delta^+ + \Delta^-$, where $\Delta^{\pm}_{x,y} :=  \delta_{x,y} \sum_z \flux^\pm_{zx} - \flux^\pm_{xy}$. Finally we define the positive definite {\it metric} $\Phi := \mathrm{diag}(\flux^+_{y \prec z})_{y \prec z}$ and the {\it gauge potential} as the $|E|\times |X|$ matrix
\bea
A_{y\prec z,x} =\left\{ \ba{ll} \flux^-_{yz}/\flux^+_{yz}, & \mathrm{if}~y = x ~\mathrm{or}~z = x \\ 0 & \mathrm{otherwise} \ea \right. . 
\eea
It is straightforward that $\Delta^+ = \partial^T \Phi \partial$, where $^T$ denotes matrix transposition. Also, $\Delta^- = \partial^T \Phi A$, since
\bea
(\partial^T \Phi A)_{x,y} = \left\{ \ba{ll} \flux_{xy}^- & \mathrm{if}~ x \neq y \\ 
\sum_{z} \flux_{yz}^- - \sum_{z} \flux_{zy}^- & \mathrm{if}~ x = y 
 \ea \right. .
\eea
We also notice that, introducing the matrix $|\partial|$ identical to $\partial$ but with all $-1$'s turned into $+1$'s (a sort of unoriented incidence matrix) we have $A = A' |\partial|$, where $A' = \mathrm{diag}\,(A_{x\prec y,y})_{x\prec y}$. This clearly expresses the fact that the covariant derivative is not a derivative\footnote{From a discrete-algebraic perspective \cite{ddg}, derivatives are boundary operators. While oriented graphs are indeed cellular complexes, i.e. a topological space with proper boundaries and boundary relationships such that ``boundaries have no boundary'', unoriented ones are not, since their boundaries might have boundary.}. Furthermore, diagonal entries of $A'$ provide an appropriate notion of ``gauge potential'' since, being defined edge-wise and being antisymmetric, they identify a discrete vector field.

Then, we can indeed put the ghost Lagrangian in the form of Eq.\,(\ref{eq:gaugemetric}). The metric, establishing a vicinity between states, is mediated by the activity, i.e. the total number of transitions that occurred between two states.

\subsection{Reflections on the gauge potential}

In Ref.\,\cite{polettini1} the Author proposed a gauge theory of Markov processes, where the underlying symmetry is the invariance of thermodynamic observables under a change of reference prior, granting the observer-independence of thermodynamics based on information theory. There, the gauge potential is a non-fluctuating quantity defined by $\ln (\pi_{xy}/\pi_{yx})$ modulo gauge transformations.  When transition probabilities have a physical origin (e.g. they obey the Arrhenius law, the mass-action law, the Kubo-Martin-Schwinger condition etc.) the Wilson loops of this gauge potential yield Clausius's expression for the entropy produced along a cyclic transformation. The candidate for a fluctuating analogue of this gauge potential would then be $\ln (\flux_{xy}/\flux_{yx})$, which is obviously different from the entries of matrix $A$ above. However, if we consider small fluxes $\flux^-_{xy} = \epsilon^-_{xy}$ we have $\flux_{xy} = \flux^+_{xy} +  \epsilon^-_{xy}$ and
\bea
\ln \frac{\flux_{xy}}{\flux_{yx}} 
= \ln \left( 1 + \frac{2\epsilon^-_{xy}}{\fluxr^+_{xy} +  \epsilon^-_{xy}} \right) \approx 2 \,\frac{\flux^-_{xy}}{\flux^+_{xy}} + O(\epsilon^2). \label{eq:loclin}
\eea
Hence when we consider small deviations from a vanishing configuration of currents we do obtain that the vector potential arising from the covariant derivative in the ghost field theory is the fluctuating version of the gauge field previously introduced in Ref.\,\cite{polettini1}.

It must be noticed that the fact that this identification is approximate might be an artifact of the discrete case, and that this discrepancy might fade away for continuous-state space diffusion processes. Indeed, the large deviation rate function for diffusion processes with empirical measure $\mu$ and current $j$ is given by \cite{maes2,wynants}
\bea
I(\mu,j) = \frac{1}{4} \int (j - j_\mu) (\mu D)^{-1}  (j - j_\mu) dx
\eea
on the assumption $\nabla j = 0$, where $j_\mu = F \mu - D \nabla \mu$,  $D$ being the diffusion matrix and $F$ the drift vector. Notice that in this expression again a fluctuating symmetric property $\mu D$, with the meaning of dynamical activity, plays the role of the metric.  The gauge theory for diffusion processes was analyzed in Ref.\,\cite{polettini4}, and there indeed the vector potential is locally linear in the currents as in Eq.\,(\ref{eq:loclin}), and not a logarithmic expression; its candidate fluctuating counterpart would then be $A = (\mu D)^{-1} j$. Therefore, generalization of our results to diffusions in continuous space, e.g. along the lines initiated in Ref.\,\cite{cooper}, might indeed lead to an elegant and exact gauge theory for finite-time fluctuations. A rigorous treatment of the small-step limit from discrete to continuous diffusion is postponed to future inquiry. Techniques based on fermionic fields, similar to the ones employed in this paper, for Fokker-Planck equations have been employed in Ref.\,\cite{supersym}.

\section{General Markov jump processes}

\label{continuous}

In this section we generalize the results of Sec.\,\ref{restless} to Markov jump processes. As we will see, application of the BEST theorem leads to no significant variation on the theme, while some care has to be paid in handling the waiting times spent at states, which enter the game in a crucial way. Below we consider two strategies to deal with them: moving to the Laplace representation (as in Refs.\,\cite{lecomte2,flomenbom}, among others), which is quite straightforward but requires an awkward (still, feasible) inversion, or considering the joint probability of the fluxes and of the so-called empirical measure. Again, we re-obtain known results from large deviation theory, putting some emphasis on certain subtle aspects related to the statistics of total number of jumps performed up to some time, that we feel have been overlooked.

\subsection{Jump processes in Laplace space}

Let $U(t) = \exp[t(W-\Omega)]$ be the propagator of a stationary ergodic continuous-time Markov semigroup on a finite state space. Entries of the off-diagonal part $W$ of the generator are the positive rates $w_{xy}$ at which transitions $x \gets y$ occur, with units of inverse time. Entries $\omega_{y} = \sum_x w_{xy}$  of the diagonal part $\Omega$ are exit rates out of states. We will work in Laplace representation, assuming it exists (the connection between semigroups, the resolvent formalism and the Laplace transform is ruled by the Hille-Yosida theorem, see \cite{hille} for a simple introduction). For all $\omega >0$, 
\bea
U(\omega) = \int e^{-t\omega I} U(t) dt  = (\omega I + \Omega - W)^{-1},
\eea
where $I$ is the identity. It can be proven that $\|W\|< \|\omega I + \Omega\|$ with respect to the usual operator norm, under which condition the following operator geometric series can be produced
\bea
(\Omega - W)^{-1} = \Omega^{-1} \sum_{N=0}^{+\infty} (W \Omega^{-1})^N. 
\eea
Making the matrix products explicit, Laplace transition amplitudes can be expressed as sums over paths
\bea
U(\omega)_{x_t,x_0} =  \sum_N \sum_{\bos{x}} P_N\prl \bos{x} \prr(\omega)
\eea
with path weight
\bea
P_N\prl \bos{x} \prr(\omega) = \frac{1}{\omega + \omega_{N}} \prod_{i < N} \frac{w_{i+1,i}}{\omega + \omega_{i}}, \label{eq:path}
\eea
where path $\bos{x} = (x_i)$ performs $N$ jumps between the constrained end-points $x_0$ and $x_N = x_t$. For notational simplicity, wherever possible we replace state $x_i$ by $i$.

The analysis then follows as for restless Markov chains. We introduce the fluxes $\bos{\flux}$ as in Eq.\,(\ref{eq:fluxes}) to obtain
\bea
P_N\prl \bos{\flux} \prr(\omega) = \frac{\varepsilon\prl \bos{\flux} \prr}{\omega + \omega_{N}} \prod_{x,y} \left(\frac{w_{xy}}{\omega + w_y}\right)^{\! \flux_{xy}} \approx \frac{\det \Delta_{x_N}\prl \bos{\flux} \prr}{\omega + \omega_{N}} \; e^{-I_N\prl \bos{\flux} \prr(\omega)} \label{eq:Pflu}
\eea
where
\bea
I_N\prl \bos{\flux} \prr(\omega) := \sum_{x,y} \flux_{xy} \ln \frac{\flux_{xy}(\omega_y + \omega)}{w_{xy} \flux_y}.
\eea
An important note is here in order. The probability of the fluxes Eq.\,(\ref{eq:Pflu}) is labelled by the total number $N$ of jumps to account for the fact that the possible values that these fluxes can take depends on $N$, since $\sum_{x,y}\flux_{xy} = N$. In mathematical terms, the filtration of the process (the time sequence of the possible events) is itself a stochastic process regulated by the probability of $N$ jumps up to time $t$. We will return on this issue later.

While Eq.\,(\ref{eq:Pflu}) is in Laplace parameter, it can already be employed in practical experiments. In fact, if instead of halting the  Markov process at a given time $t$ we conduct experiments in variable time sampled with probability $P(t) = \omega e^{- \omega t}$, then the probability distribution of the fluxes for such protocol is given by $P_{N,\omega} \prl \bos{\flux} \prr = \omega \, P_N\prl \bos{\flux} \prr(\omega)$. In this case, since typical trajectories halt at time $\omega^{-1}$ and prolonging this time makes the probability smaller, large fluctuations are very rare and it is necessary to retain the full structure of the probability density function.

\subsection{All-time probability of the fluxes}

If instead we insist on performing experiments in a given time, we need to perform the inverse Laplace transform. Isolating the explicit dependency on $\omega$ we have
\bea
P_N\prl \bos{\flux} \prr(t) =  \det \Delta_{x_t}\prl \bos{\flux} \prr \; e^{-I_N\prl \bos{\flux} \prr(0)}\mathcal{L}^{-1} \left[ \prod_x \frac{1}{(1 + \omega/\omega_x)^{\flux_x}}\right](t). 
\eea
The inverse Laplace transform to be calculated is given by
\bea
\mathcal{L}^{-1} \left[ \prod_x \frac{1}{(\omega_x + \omega)^{\flux_x}}\right](t) = \sum_x Q_x(t) e^{-\omega_{x} t} 
\eea
where $Q_x(t)$ is a polynomial of degree $\flux_x -1$ in $t$. The general idea is that $\prod_x (\omega_{x} + \omega)^{\flux_x}$ is the common denominator for denominators $(\omega_{x} + \omega)^{m_x}$, for all $n_x = 0,\ldots,\flux_x-1$. If there exist coefficients $C_x^{n_x}$ such that 
\bea
\prod_x \frac{1}{(\omega_{x} + \omega)^{\flux_x}} = \sum_x \sum_{n_x =1}^{\flux_x} C_x^{n_x-1} \frac{1}{(\omega_{x} + \omega)^{m_x}},
\eea
given that the Laplace transform of these latter terms is known,
\bea
\mathcal{L}^{-1} \left[ \frac{1}{(\omega_{x} + \omega)^{n_x}} \right](t) = \frac{t^{n_x-1}}{(n_x-1)!} e^{- \omega_{x} t}
\eea
we obtain
\bea
Q_x(t) = \sum_{m_x =0}^{\flux_x-1} C_x^{n_x} \frac{t^{n_x}}{n_x!}.
\eea
The task is then to find the coefficients $C_x^{m_x}$. We derive the following formula in \ref{appB},
\bea
C_x^{n_x} & = & (-1)^{\flux_x -1 - n_x} \hspace{-0.4cm}
 \sum_{\substack{\{m_y\}\\  \sum_{y\neq x} m_y \\ ~~= \flux_x - n_x - 1}}  \hspace{-0.4cm} \prod_{y \neq x}  \frac{1}{( \omega_{y} - \omega_x)^{\flux_y + m_{y}}} \hspace{-0.4cm}  \sum_{\substack{\{m_{q}\}\\ \sum_{q = 1}^{\flux_x -1} q m_{q} \\ ~~~= m_y}}  \hspace{-0.4cm} \prod_{q = 1}^{\flux_x -1} \frac{(\flux_{y}/q)^{m_{q}}}{m_{q}!}\nonumber \\
\eea
Though not very practical, this formula is programable in software simulations. 

\subsection{Joint fluxes and empirical measure}

The empirical measure accounts for the resting time at states. In this section we derive an exact formula for the joint statistics of the fluxes and of the empirical measure at time $t$, for a given number of jumps $N$.

Before jump $x_i \to x_{i+1}$, the process waits an interval of time $t_i$ at state $x_i$, for a total elapsed time $\sum_i t_i = t$. We then introduce the empirical measure as the fraction of time spent at $x$:
\bea
\tau_x[\bos{x},\bos{t}]  := \sum_{i} t_i  \,\delta_{i,x} .
\eea
Notice that $\sum_x \tau_x[\bos{x},\bos{t}]  = \sum_{i} t_i  = t$.  We can then express Eq.\,(\ref{eq:path}) as
\bea
P_N\prl \bos{x} \prr(\omega) & = & \int \mathcal{D} \bos{t} \; e^{-(\omega + \omega_N) \tau_N } \prod_{i < N}  w_{i+1,i} \, e^{-(\omega_{i} + \omega)t_i} \\
& = &  \int dt \, e^{- \omega t}  \int \mathcal{D} \bos{t} \, P_N\prl \bos{x},\bos{t} \prr(t)
\label{eq:lala}
\eea
where all integrals range from $0$ to $+\infty$ and 
\bea
P_N\prl \bos{x},\bos{t} \prr(t) & := & \delta\Big(t - \sum_i  t_i \Big) e^{-\omega_N t_N} \prod_{i < N}  w_{i+1,i} \, e^{-\omega_{i} t_i} \\
& = & \delta\Big(t - \sum_x  \tau_x[\bos{x},\bos{t}] \Big) e^{- \sum_x \omega_{x} \tau_x[\bos{t},\bos{x}]} \prod_{x,y} w_{xy}^{\flux_{xy}[\bos{x}]}.
\eea
is the probability that up to time $t$, given $N$ jumps, trajectory $\bos{x},\bos{t}$ has been performed. We now want to marginalize for $\bos{\flux},\bos{\tau}$. To marginalize out the waiting times, noticing that the above expression only depends on $\bos{\tau}[\bos{x},\bos{t}]$, we need to calculate
\bea
\int \mathcal{D}\bos{t} \, \prod_x \delta_{\tau_x[\bos{x},\bos{t}],\tau_x}
= \prod_x \frac{\tau_x^{\flux_x}}{\flux_x!}.
\eea
Notice that for large enough $\bos{\flux}$ the factorial $\flux_x!$ is counterbalanced by that coming from the Eulerian tour counting. We then obtain, using Stirling, 
\bea
P_N\prl \bos{\flux},\bos{\tau} \prr(t) 
 =  \delta\Big(t - \sum_x  \tau_x\Big)  \det \Delta_{x_N} \exp - I_N\prl \bos{\flux},\bos{\tau} \prr \label{eq:N}
\eea
where
\bea
I_N\prl \bos{\flux},\bos{\tau} \prr := \sum_{x,y} \flux_{xy} \ln \frac{\flux_{xy}}{w_{xy}\tau_y} + \sum_x \tau_x \, \omega_{x}. \label{eq:RFMP}
\eea

\subsection{Large deviations of the fluxes}

\label{largecont}

When considering large deviations one requires that the fluxes and the empirical measure grow linearly in time
\bes
\bos{\flux} & \sim & t \dot{\bos{\flux}} \\
\bos{\tau} & \sim & t \dot{\bos{\tau}},
\ees
where the dot serves as a notation. The joint rate function for the fluxes and the empirical measure has been discussed in Ref.\,\cite{maes2} and, interestingly, it can be further contracted for the joint currents and empirical measure  \cite{wynants,bertini2}. As a small improvement we can work out an expression for the large deviation rate function of the fluxes. We define
\bea
\dot{I}_N(\cdot) := t^{-1}I_N(t\,\cdot).
\eea
Notice that we retain the explicit dependence on $N$ for reasons that will be explained in the following section. Let us take the variation of $\dot{I}_N$ with respect to $\dot{\bos{\tau}}$, accounting for the normalization $\sum_x \dot{\tau}_x=1$ by means of a Lagrange multiplier $\lambda$. Setting $\delta/ \delta \dot{\tau}_x \left[\dot{I}_N + \lambda (\sum_{y} \dot{\tau}_y -1) \right] = 0$  yields $\dot{\tau}_x^\ast = \dot{\flux}_x/(\dot{\omega}_x + \lambda)$, where the multiplier satisfies
\bea
\sum_{x}  \frac{\dot{\flux}_x}{\omega_x + \lambda[\dot{\bos{\flux}}]} = 1. \label{eq:determined}
\eea
We then obtain
\bea
\dot{I}_N\prl \dot{\bos{\flux}} \prr = \sum_{x,y} \dot{\flux}_{xy} \ln \frac{\dot{\flux}_{xy}(\omega_y +  \lambda[\dot{\bos{\flux}}])}{w_{xy} \dot{\flux}_y} + \sum_x   \frac{\dot{\flux}_x \omega_x}{\omega_x +  \lambda[\dot{\bos{\flux}}]}.
\eea
This, together with the constraint Eq.\,(\ref{eq:determined}), provides an expression for the large deviation rate function for the fluxes, with an important {\it caveat} that we next discuss.

\subsection{Total number of jumps and large deviations}

Derivations of the large deviation rate function for the fluxes and the empirical measure often employ the smart trick of comparing the probability of a fluctuation with that of a system where that configuration is typical. The rate function found by the above mentioned authors is (see Eq.\,(11.7) in Ref.\,\cite{wynants}):
\bea
J(\dot{\bos{\flux}},\dot{\bos{\tau}}) & = &  \sum_{x,y} \dot{\flux}_{xy} \ln \frac{\dot{\flux}_{xy}}{w_{xy} \dot{\tau}_y} + \sum_x \dot{\tau}_x \, \omega_{x} - \sum_{x,y} \dot{\flux}_{x,y} \\
& = & \dot{I}_N(\dot{\bos{\flux}},\dot{\bos{\tau}}) - \frac{N}{t} \label{eq:constant}
\eea
The second line identifies $J$ with our expression but for an additional term  $- N/t$. From our perspective, the origin of this term is subtle and we haven't seen it treated elsewhere. In the following we attempt a discussion.

First we notice that there is no such issue in the case of restless Markov chains, because $ N=t$ with certainty and additive constants to the rate function are irrelevant. For jump processes, the total number of jumps $N$ is a stochastic variable itself and it obeys a large deviation principle. It might be that at some time $t$ a number $N \propto t$ of jumps have occurred, but (since we are interested in large deviations!) we need to keep into account that also $N$ might fluctuate. In other words, the two extensive parameters $N$ and $t$ are decoupled and one needs to specify how one scales with the other. Importantly, while $\bos{\phi}$ determines $N$, its statistics depends on $N$ in a subtle way. In fact, a Markov jump process is built by superimposing  Poisson processes on top of a Markov chain \cite{ethier}, so that the sequence of measurable values that the fluxes can take (the so-called {\it filtration} of the process) is itself a stochastic process.

The simpler way out of this riddle is to just suppose that $N = \omega^\ast t$ independently attains the minimum of its large deviation function, that is, that we look at the most probable fluxes assuming that the total number of jumps up to some time is the most probable one. Then in Eq.\,(\ref{eq:constant}) we obtain an additional constant that does not affect the large deviations. {\it A posteriori}, it is intuitive that the most probable configuration of fluxes should be such that, of all the jumps out of $y$, the ratio of those that go to $x$ should equal to the probability that the next jump out of $y$ is to $x$:
\bea
\frac{\dot{\flux}_{xy}^\ast}{\dot{\flux}_{y}^\ast} = \frac{w_{xy}}{\omega_y}.
\eea
Notice that $w_{xy}/\omega_y$ are transition probabilities of a restless Markov chain. Then $\dot{\flux}_x^\ast  = \omega^\ast \rho_x$, $\bos{\rho}$ being the steady state of the restless Markov chain. One can easily check that this is the stationary solution corresponds to  $\lambda = 0$ and
\bea
\omega^\ast = \left( \sum_x \frac{\rho_x}{\omega_x} \right)^{-1}.
\eea
As a consistency check, notice that these are indeed the exact expressions when one takes all identical exit rates $\omega_x = \omega^\ast$. Therefore, we do obtain the usual rate function as in Refs.\,\cite{maes2,wynants,bertini2}, when the total number of jumps is assumed to attain its most probable value. However, the assumption is quite a strong one, as it has been shown that there exist kinetically constrained models where the total number of jumps undergoes a dynamical phase transition between an active phase with $N = \omega^\ast t$ and an inactive phase with $N$ subextensive in time \cite{nemoto2,bodineau}.

We hold that a more complete and correct treatment should instead keep into account that the total number of jumps fluctuates, and that this might imply modifications of the formula usually found in the literature. The large deviation rate function of the total number of jumps has been the subject of a prolific branch of studies \cite{nemoto2,bodineau}. For the case of all identical exit rates $\omega_x = \omega^\ast$, the probability of $N$ total jumps is Poissonian with average $\omega^\ast t$
\bea
P\prl N \prr(t)  = e^{-\omega^\ast t}\frac{(\omega^\ast t )^N}{N!}
\eea
with rate function
\bea
\dot{I}\prl\dot{N} \prr = \omega^\ast + \dot{N} \ln \frac{\dot{N}}{\omega^\ast}.
\eea
This piece of information should somehow feed back into the rate function of the fluxes if the general question to be asked is what is the probability of a configuration of fluxes, regardless of the total number of jumps. We demand this task to future inquiry. Departure from Poissonian statistics in a low-density fluid, treated with linearized Boltzmann equation, is found in Ref.\,\cite{visco}. Analogies and differences between large $N$ and large $t$ ensembles have been studied in Ref.\,\cite{budini}.

\section{Conclusions}

\label{conclusions}

In this paper we established a fascinating connection between Large Deviation Theory, combinatorics, and certain techniques and concepts from the Quantum Field Theory of Yang-Mills theories. The main result is a generalization of known results \cite{maes2,wynants,bertini2} regarding the large deviation rate function of the fluxes and of the empirical measure for Markov processes  (level 2.5 large deviations) to the full probability distribution. The result relies on a slight generalization of the so-called BEST theorem in graph combinatorics to Eulerian tours with open ends. On the other hand, since Eulerian tours can be counted by spanning trees, recently proposed methods to express matrix-tree determinants as Fermionic field theories are applicable. The Lagrangian so found displays an interaction of the ghost field with an underlying metric, representing the activity of the system, and with a gauge potential that is related to the nonequilibrium gauge potential that the Author explored in a previous publication \cite{polettini4}.

The mathematics of spanning trees is very rich and full of connections, e.g. to knot theory via the Alexander's polynomial \cite{enumeration}, to spin networks in Loop Quantum Gravity, to Feynman diagrammatics. A simple open question that already emerges is whether our slight generalization of the BEST theorem to quasi-balanced digraphs might generalize to {\it any} digraph, by a suitable decomposition in terms of several Eulerian tours with open ends. Finally, establishing the statistics of the currents at all times might play an important role in the study of finite-time thermodynamics.

An open problem that will deserve further attention is that of the apparent dependence of the large deviation function, as we derived it, on a given total number of jumps and, more in general, a more careful characterization of the statistics of the total number of jumps, which is affected by the stochastic trajectory and in turn determines the set of measurable events.

\section*{Acknowledgments}

The research was supported by the National Research Fund of Luxembourg in the frame of Project No. FNR/A11/02 and the AFR Postdoc Grant 5856127.

\appendix 

\section{BEST theorem for quasi-balanced digraphs}

\label{appA}

Let $(X,\bos{\flux})$ be a connected balanced digraph with vertex set $X$. Let $\flux_{xy}$ be the number of arrows $y \to x$ and $\flux^{\mathrm{\,out}}_x = \sum_y \flux_{yx}$ be the out-degree of vertex $x$. 

\begin{definition}
A \emph{quasibalanced digraph} is a digraph such that one vertex $x_{\mathrm{init}}$ has one more outgoing arrow than ingoing ones, one vertex $x_{\mathrm{fin}}$ has one more incoming arrow than outgoing ones, and all other vertices are balanced.
\end{definition}

\begin{definition}
An \emph{open Eulerian tour} on a quasibalanced digraph is a path that starts at $x_{\mathrm{init}}$, ends at $x_{\mathrm{fin}}$ and traverses each arrow exactly once.
\end{definition}

\begin{theorem}
A connected quasibalanced digraph admits an open Eulerian tour.  
\end{theorem}

\begin{proof}
By adding an arrow $x_{\mathrm{init}} \gets x_{\mathrm{fin}}$ one obtains a connected balanced digraph, which admits a closed Eulerian tour starting and ending at $x_{\mathrm{init}}$. By cyclically permuting the arrows in this tour one can obtain a Eulerian tour whose last arrow is $x_{\mathrm{init}} \gets x_{\mathrm{fin}}$. By removing this arrow one obtains an open Eulerian tour on the quasibalanced digraph.
\end{proof}

\begin{theorem}
\label{th2}
The number of open Eulerian tours on a connected quasibalanced digraph is given by
\bea
E = T_{x_{\mathrm{fin}}} \prod_x \left(\flux^{\mathrm{\,out}}_x - 1 + \delta_{x,x_{\mathrm{fin}}}\right)!
\eea
where $T_{x_{\mathrm{fin}}}$ is the number of oriented spanning trees with root in $x_{\mathrm{fin}}$.
\end{theorem}

\begin{proof}
We follow the arguments found in Stanley's book \cite{stanley}, pp. 54-57, for balanced digraphs, with slight modifications. Consider an open Eulerian tour, that by the previous theorem exists. For all vertices $x\neq x_{\mathrm{fin}}$ let $e(x)$ be the last arrow that exits $x$. Then the set $T = \{e(x),x\neq x_{\mathrm{fin}}\}$ is an oriented spanning tree with root in $x_{\mathrm{fin}}$ (see \cite{stanley}, Proof of Claim 1). Conversely, given a spanning tree with root in $x_{\mathrm{fin}}$ one can construct open Eulerian tours starting at $x_{\mathrm{init}}$ and picking random consecutive edges that do not belong to the spanning tree, unless there are no edges left out, in which case one will pick an edge $e(x)\in T$ and exit $x$ for the last time.  We need to show that by this procedure one doesn't get stuck at $x_{\mathrm{fin}}$ before traversing all other arrows, that is, that we do not need to ``lift the pencil''. If that was the case, then it means that at least a simple cycle of arrows $C$ (with no crossings or multiple edges) has been left out of the tour $\gamma$ so far generated (since $E \setminus \gamma$ is a balanced digraph). Since the spanning tree $T$ does not contain cycles, one of these arrows $e \in C$ out of a vertex $y$ is not in the spanning tree, and $e(y)$ does not belong to the cycle. Then either we exited $y$ by $e(y)$ before using $e$, which is in contradiction with the above procedure, or $e(y)$ has not been used, which implies that it belongs to another simple cycle $C'$ that has been left out of $\gamma$. Along this other simple cycle there exists an edge $e' \notin T$, and one can continue the reasoning to obtain a contradiction. Then there are  $\prod_x \left(\flux^{\mathrm{\,out}}_x - 1 + \delta_{x,x_{\mathrm{fin}}}\right)!$ ways to pick such random arrows, with the $-1$ accounting for the edges in the spanning tree, but for vertex $x_{\mathrm{fin}}$. Since every choice of a tree $T$ and every random selection of remaining arrows generates a different Eulerian tour, and every Eulerian tour can be so described, we can conclude.  
\end{proof}

\section{An inverse Laplace transform} 

\label{appB}

The task is to find the inverse Laplace transform of
\bea
F(\omega) = \prod_x \frac{1}{(\omega + \omega_x)^{\flux_x}}.
\eea
The result is a generalization of Heaviside's expansion theorems \cite{abrasteg}. We employ the well-known residue formula
\bea
f(t) := \mathcal{L}^{-1} \left[F(\omega)\right](t)  =  \sum_x \mathrm{Res}\, \left[ F(\omega) e^{\omega t} ,-\omega_x \right].
\eea
Since $-\omega_x$ is a pole of degree $\flux_x$, we have
\bea
f(t) = \sum_x \frac{1}{(\flux_x -1)!} \lim_{\omega \to -\omega_x} G^{(\flux_x-1)}_x(\omega) \label{eq:inverselap}
\eea
where we introduced
\bea
G^{(p)}_x(\omega) := \frac{d^p}{d\omega^p} \Big[ e^{\omega t} \prod_{y\neq x} (\omega + \omega_{y})^{-\flux_{y}} \Big].
\eea
Interpreting $G^{(0)}_k$ as a moment generating function, instead of directly evaluating moments it's simpler to evaluate the cumulants
\bea
H^{(q)}_x(\omega) & := & \frac{d^q}{d\omega^q} \ln G^{(0)}_x(\omega) \\
& = & t\, \delta_{q,1}  + (-1)^q (q-1)!  \sum_{y\neq x}\frac{\flux_{y}}{(\omega + \omega_{y})^q} \\
& =: & t\, \delta_{q,1} + \bar{H}^{(q)}_x(\omega).
\eea
The generalization of the chain rule to higher derivatives is given by Fa\`a di Bruno's formula, that simplifies for the exponential function $G^{(0)}_x(\omega) = \exp H^{(0)}_x(\omega)$:
\bea
G^{(p)}_x(\omega) & = & G^{(0)}_x(\omega)\,\sum \frac{p!}{m_1!\ldots m_p!} \prod_{q = 1}^p \left[ \frac{H^{(q)}_x(\omega)}{q!}\right]^{m_q}\\
& = &  G^{(0)}_x(\omega) B_p\left(H^{(1)}_x(\omega),\ldots,H^{(p)}_x(\omega)\right)
\eea
where the sum runs over all $m_1,\ldots,m_p \geq 0$ such that
\bea
\sum_{q=1}^p q m_q = p.
\eea
The second line identifies the complete Bell polynomials. Let us now isolate powers of $t$ using of the binomial formula, and rearrange summations. After some work we obtain
\bea
G^{(p)}_x(\omega) & = & G^{(0)}_x(\omega)\,\sum \left[t + \bar{H}^{(1)}_x(\omega) \right]^{m_1} \frac{p!}{m_1!\ldots m_p!} \prod_{q = 2}^p \left[\frac{ \bar{ H}^{(q)}_x(\omega)}{q!}\right]^{m_q} \nonumber \\
& = & G^{(0)}_x(\omega)\,\sum_{\{m_q\}} \sum_{n_x=0}^{m_1} \frac{m_1!}{n_x! (m_1-n_x)!} t^{n_x} \left[ \bar{H}^{(1)}_x(\omega) \right]^{m_1-n_x} \frac{p!}{m_1!\ldots m_p!} \prod_{q = 2}^p \left[\frac{ \bar{ H}^{(q)}_x(\omega)}{q!}\right]^{m_q} \nonumber \\
& = & G^{(0)}_x(\omega)\, \sum_{n_x=0}^{p}  \frac{t^{n_x}}{n_x!} \sum_{\{\bar{m}_q\}} \frac{p!}{\bar{m}_1!\ldots \bar{m}_p!} \prod_{q = 1}^p \left[\frac{ \bar{ H}^{(q)}_x(\omega)}{q!}\right]^{\bar{m}_q} 
\eea
where we have defined $\bar{m}_1 = m_1 - n_x$, $\bar{m}_{j>1} = m_{j>1}$ and now the sum over the $\{\bar{m}_j\}$ is constrained by
\bea
\sum_{q=1}^p q \bar{m}_q = p - n_x.
\eea
Finally
\bea
f(t) = \sum_x e^{-\omega_x t} \sum_{n_x=0}^{\flux_x -1}  \frac{t^{n_x}}{n_x!} C_x^{n_x}
\eea
with 
\bea
C_x^{n_x} = (-1)^{\flux_x -1 - n_x} \prod_{y\neq x} (\omega_{y}-\omega_x)^{-\flux_{y}}
 \sum_{\{\bar{m}_q\}} \prod_{q = 1}^{\flux_x -1} \frac{1}{\bar{m}_q!} \left[
 \frac{1}{q}  \sum_{y\neq x}\frac{\flux_{y}}{( \omega_{y} - \omega_x)^q} 
\right]^{\bar{m}_q}
\eea
We can further introduce a multinomial expansion with $\sum_{y\neq x} m_{q,y} = m_q$, define $m_y = \sum_q q m_{q,y}$, replace $\sum_{\{\bar{m}_q\}}  \sum_{\{m_{q,y}\}} $ by $ \sum_{\{m_y\}}  \sum_{\{m_{q,y}\}}$, keeping into account the constraints $\sum_{y\neq x} m_y = \flux_x - n_x - 1$ and $\sum_{q = 1}^{\flux_x -1} q m_{q,y} = m_y$, to obtain
\bea
C_x^{n_x} & = & (-1)^{\flux_x -1 - n_x} \prod_{y\neq x} (\omega_{y}-\omega_x)^{-\flux_{y}}
 \sum_{\{\bar{m}_q\}}  \sum_{\{m_{q,y}\}}   \prod_{q = 1}^{\flux_x -1} \prod_{y \neq x} \frac{1}{m_{q,y}!} \frac{(\flux_{y}/q)^{m_{q,y}}}{( \omega_{y} - \omega_x)^{q m_{q,y}}} \nonumber \\
 & = & (-1)^{\flux_x -1 - n_x} 
 \sum_{\{m_y\}}\frac{1}{\prod_{y \neq x}( \omega_{y} - \omega_x)^{\flux_y + m_{y}}}  \sum_{\{m_{q,y}\}}  \prod_{q = 1}^{\flux_x -1}  \prod_{y \neq x} \frac{(\flux_{y}/q)^{m_{q,y}}}{m_{q,y}!}. 
\eea

\section{Example} 

\label{app3}

Consider as the state space the following graph:
\bea
\xymatrix{ a  \ar@{-}[r] & b \ar@{-}[d] \\ 
d \ar@{-}[ur] \ar@{-}[u] & c \ar@{-}[l]}
\eea
A configuration of fluxes identifies the following quasi-balanced digraph:
\bea
\xymatrix{a  \ar@/^{3pt}/[r]   \ar@/_{3pt}/[r]  & b  \ar[d]  
\ar@/^{4pt}/[dl]  \ar[dl]   \\ 
d  \ar[u]   \ar@/^{4pt}/[ur]  & c   \ar[l]   }
\eea
The Laplacian matrix reads
\bea
\Delta = \left(\ba{cccc} 2 & 0 & 0 & -1 \\ -2 & 3 & 0 & -1 \\ 0 & -1 & 1 & 0 \\ 0 & -2 & -1 & 2 \ea\right).
\eea
Let us enumerate the 6 independent Eulerian paths going from the source $a$ to the sink $d$ (to do so systematically one can proceed as follows: start at $a$, and every time choices ramify into several possible directions, choose the one with the lowest letter until all transitions are exhausted; then choose the second lowest at the latest ramification, and so on\ldots):
\bea
abcdabdbd \nonumber \\
abcdbdabd \nonumber \\
abdabcdbd \nonumber \\
abdabdbcd \nonumber \\
abdbcdabd \nonumber \\
abdbdabcd
\eea
Hence, accounting for the Boltzmann factorials for the fluxes, we obtain $6 \times 2! \times 2! = 24$ Eulerian tours. 

There are two spanning tree pointing at $d$:
\bea
\ba{c}
\xymatrix{a  \ar[r] & b   \ar[d]
  \\ 
d   & c \ar[l] }\ea \qquad \qquad \ba{c} \xymatrix{a  \ar[r] & b 
\ar[dl]  \\ 
d   & c \ar[l] } \ea
\eea
The first occurs $2$ times because of the choice of edge $a \to b$, the second $4$ times because of the choice of edges $a \to b$ and $b \to d$, yielding $6$ spanning trees (that this number coincides with the number of independent paths is a mere coincidence). Finally we have that this number has to be multiplied by the factorials at the vertices, $(3-1)!$ at $b$ and $2!$ at $d$, yielding $24$ Eulerian tours.

The cofactor reads
\bea
\det \Delta_{d} = \left| \ba{ccc} 2 & 0 & 0  \\ -2 & 3 & 0  \\ 0 & -1 & 1 \ea\right|  = 6
\eea
and indeed it yields the number of oriented spanning tree rooted at $d$.